\documentclass[12pt]{amsart}

\usepackage{amsfonts}
\usepackage{amsmath}
\newtheorem{thm}{Theorem}[section]

\newtheorem{prop}[thm]{Proposition}
\newtheorem*{prob*}{Problem}
\newtheorem*{thm*}{Theorem}

\theoremstyle{definition}
\newtheorem{defn}[thm]{Definition}

\newtheorem*{defn*}{Definition}

\newtheorem*{rem*}{Remark}
\numberwithin{equation}{section}

\newcommand{\C}{\mathbb C}
\newcommand{\R}{\mathbb R}

\DeclareMathOperator{\Bessel}{Bessel}

\newcommand{\Tr}{\mathop{\mathrm{Tr}}}

\begin{document}
\title[Differential  equations for singular values]
 {\bf{Differential equations for singular values of products of Ginibre random matrices}}
\author{Eugene Strahov}
\address{Department of Mathematics, The Hebrew University of
Jerusalem, Givat Ram, Jerusalem
91904}\email{strahov@math.huji.ac.il}


\commby{}
%%%%%%%%%%%%%%%%%%%%%%%%%%%%%%%%%%%%%%%%%%%%%%%%%%%%%%%%%%
%%%%%%%%%%%%%%%%%%%%%%%%%%%%%%%%%%%%%%%%%%%%%%%%%%%%%%%%%%
\begin{abstract}
It was proved by Akemann, Ipsen and Kieburg \cite{AkemannIpsenKieburg}  that squared singular values of products of $M$ complex Ginibre random matrices form a determinantal point process whose correlation
kernel  is expressible in terms of Meijer's $G$-functions. Kuijlaars and Zhang \cite{Kuijlaars} recently  showed  that at the edge of the spectrum, this correlation kernel has a remarkable scaling limit $K_M(x,y)$ which can be understood as a generalization of the classical Bessel kernel of Random Matrix Theory.
In this paper we investigate the Fredholm determinant of the operator with the kernel
$K_M(x,y)\chi_J(y)$, where $J$ is a disjoint union of intervals, $J=\cup_j(a_{2j-1},a_{2j})$, and $\chi_J$ is the characteristic function of the set $J$.
This Fredholm determinant is equal to the probability that $J$ contains no particles of the limiting determinantal point process defined by $K_M(x,y)$ (the gap probability).
We derive  Hamiltonian differential equations  associated with the corresponding Fredholm determinant, and relate them with the monodromy preserving deformation equations
of the  Jimbo, Miwa, M$\hat{\mbox{o}}$ri, Ueno and Sato theory. In the special case $J=(0,s)$ we give a formula for the gap probability in terms of a solution of a system of non-linear ordinary differential equations.
\end{abstract}
\maketitle
\section{Introduction}
It is a well-known fact that a possible  language  for a description
of probabilistic quantities of interest in the theory of exactly
solvable random matrix models is the language of non-linear and
partial differential equations. We refer the reader to the works of
Tracy and Widom \cite{TracyWidomIntroduction}-\cite{TracyWidom},
Adler, Shiota and van Moerbeke \cite{AdlerShiotaMoerbeke,
AdlerShiotaMoerbeke1}, to the surveys by van Moerbeke
\cite{Moerbeke1, Moerbeke2}, and to the book by Forrester
\cite{ForresterLogGases} for an introduction to this aspect of
Random Matrix Theory, and for main results in this area of research.

One classical  example of such a description is that of singular values of a random complex  Ginibre matrix. Alternatively, one can think about eigenvalues of a complex Wishart
matrix, i.e. about eigenvalues of a matrix $X^*X$, where $X$ is a random complex Ginibre matrix of size $N\times K$.
By definition, a random complex Ginibre matrix is a rectangular matrix whose entries are independent standard complex Gaussian variables. It turns out that the squared singular values of a complex Ginibre matrix form a determinantal process on $(0,\infty)$ (called the classical Laguerre ensemble). A remarkable feature of this determinantal point process is that its correlation kernel can be written explicitly in terms of the classical Laguerre polynomials. This allows to study different asymptotic regimes of the classical Laguerre ensemble, see, for example, Forrester \cite{ForresterLogGases}, Chapter 7.  When one rescales the classical Laguerre  ensemble at the hard edge of the spectrum the limiting determinantal point process with the kernel
 \begin{equation}\label{BesselKernel}
K_{\Bessel}(x,y)=\frac{J_{\nu}(\sqrt{x})\sqrt{y}J_{\nu}'(\sqrt{y})-\sqrt{x}J_{\nu}'(\sqrt{x})J_{\nu}(\sqrt{y})}{2(x-y)}
\end{equation}
arises. (Here $J_{\nu}(z)$ stands for the Bessel function of the first kind of order $\nu$, and $\nu=N-K$, see \cite{Olver}, formula 10.9.23). Moreover, different probabilistic quantities of interest are expressible in terms of Fredholm determinants of operators defined by this kernel. For example, the Fredholm determinant of the operator with the kernel $K_{\Bessel}(x,y)\chi_J(y)$, where $J$ is a disjoint union of intervals, $J=\cup_j(a_{2j-1},a_{2j})$, and $\chi_J$ is the characteristic function of the set $J$ can be understood as the gap probability. This is the probability of the event
that there are no particles of the limiting determinantal process in $J$.
The theory of the Fredholm determinant defined by the Bessel kernel  is  developed by Tracy and Widom \cite{Tracy}. Namely, Tracy and Widom \cite{Tracy} obtain a system of partial differential equations for the logarithmic derivative of this Fredholm determinant. These partial differential equations admit a Hamiltonian formulation in which the end points $a_j$ of the intervals play the role of multi-time variables. Subsequently, it was shown in Palmer \cite{Palmer}, Harnad, Tracy, and Widom \cite{HarnadTracyWidom}, Harnad \cite{Harnad}  that the partial differential equations characterizing the Fredholm determinant of the operator with the kernel $K_{\Bessel}(x,y)\chi_J(y)$ can be viewed as a special case  of the monodromy preserving deformation equations of the Jimbo, Miwa, M$\hat{\mbox{o}}$ri, Ueno and Sato theory \cite{Jimbo}-\cite{Jimbo2}. In this framework, the Fredholm determinant of the operator defined by the Bessel kernel can be understood  as the tau-function, and the analysis
of the Fredholm determinant can be considered as a special case of the analysis of tau-functions.

In the case of a single  interval Tracy and Widom \cite{Tracy} give a representation of the Fredholm determinant in terms of a solution of the Painlev\'{e} V equation.
This allows to understand the asymptotic behavior of the probability of the event that there is no particle of the limiting determinantal point process in $(0,s)$, as $s\rightarrow\infty$.

It is the aim of the present paper to extend some of the results mentioned above to \textit{products of independent Ginibre matrices}. Such products arise in very different
areas of research, see, for example, M$\ddot{\mbox{u}}$ller \cite{Muller1,Muller2}, Akemann, Ipsen, and Kieburg \cite{AkemannIpsenKieburg}  for applications in the theory of telecommunications. In the context of this paper  the most important fact is that products of independent Ginibre matrices lead to determinantal point processes both in the complex plane $\C$ and on
the real line $\R$.
This has been shown recently by Akemann and Burda \cite{Akemann1},
Akemann,  Kieburg and Wei \cite{AkemannKieburgWei},
 Akemann, Ipsen and Kieburg \cite{AkemannIpsenKieburg}, Ipsen and Kieburg \cite{IpsenKieburg} (see also
Adhikari,  Reddy,  Reddy, and Saha \cite{AdhikariReddyReddySaha}).
 The correlation kernels of such determinantal processes can be expressed in terms of Meijer's $G$-functions, which enables to  investigate the statistics of eigenvalues and of singular values for products of independent Ginibre matrices by the usual methods of Random Matrix Theory. We refer the reader to papers by  Akemann and Strahov \cite{AkemannStrahov}, Zhang \cite{Zhang}, Kuijlaars and Zhang \cite{Kuijlaars}  (in addition to papers just mentioned above) for some recent results in this direction.
In particular, it was proved by Akemann, Ipsen and Kieburg
\cite{AkemannIpsenKieburg}  that squared singular values of products
of $M$ complex Ginibre random matrices form a determinantal point
process on $\R_{>0}$ whose correlation kernel  is expressible in
terms of Meijer's $G$-functions. Kuijlaars and Zhang
\cite{Kuijlaars}   show  that at the edge of the spectrum, this
correlation kernel has a remarkable scaling limit $K_M(x,y)$ which
can be understood as a generalization of the classical Bessel kernel
of Random Matrix Theory. In this paper we investigate the Fredholm
determinant of the operator with the kernel $K_M(x,y)\chi_J(y)$,
where $J$ is a disjoint union of intervals,
$J=\cup_j(a_{2j-1},a_{2j})$, and $\chi_J$ is the characteristic
function of the set $J$. We derive  Hamiltonian differential
equations  associated with the corresponding Fredholm determinant,
and relate them with the monodromy preserving deformation equations
of the Jimbo, Miwa, M$\hat{\mbox{o}}$ri, Ueno and Sato
\cite{Jimbo}-\cite{Jimbo2} theory. In the special case $J=(0,s)$ we
give a formula for the gap probability in terms of a solution of a
system of non-linear ordinary differential equations.

This paper is organized as follows. In Section
\ref{SubsectionSingularvalues} we summarize  exact and asymptotic
results on singular values   of products of Ginibre random matrices
which are relevant for this work. In particular, in Section
\ref{SubsectionSingularvalues}  we present an explicit formula for
the limiting correlation kernel $K_M(x,y)$ found in Kuijlaars and
Zhang \cite{Kuijlaars}. In Section \ref{SectionMainResults} we state
new results obtained in this paper. Proposition  \ref{Proposition
Partial Differential Equations for x y xi eta} gives a system of
partial differential equations associated with  the Fredholm
determinant of the operator with the kernel $K_{M}(x,y)\chi_J(y)$.
Proposition \ref{PropositionHamiltonianEquations} provides a
Hamiltonian formulation of these partial differential equations, and
in Section \ref{SubsectionIsomonodromicSystem} we relate them with
isomonodromic deformation equations, see Proposition
\ref{PropositionSchlesinger} and Proposition
\ref{PropositionIsomonodromyFormulation}. The special case $J=(0,s)$
is considered in in Section \ref{SubsectionSpecialCase}, and
Proposition \ref{PropositionNonlinearOrdinaryDiffEquation} gives a
formula for the probability that no particles of the determinantal
process lie in the interval $(0,s)$. Section \ref{SectionProofs} is
devoted to the derivations of differential equations and to proofs
of other results stated in Section \ref{SectionMainResults}.
Finally, Section \ref{SectionConclusions} summarizes the results of
this work in a non-technical language.

 \section{Singular values of products of random complex matrices}\label{SubsectionSingularvalues}
 To present the results on singular values of products of random complex matrices
 let us adopt the same notation and definitions
for Meijer's $G$-function as in Luke \cite{Luke}, Section 5.2.
Namely, the Meijer G-function
$G_{p,q}^{m,n}\left(x\biggl|\begin{array}{cccc}
                                                                a_1, & a_2, & \ldots, & a_p \\
                                                                b_1, & b_2, & \ldots, & b_q
                                                              \end{array}
\right)$
is defined as
\begin{equation}
\begin{split}
G_{p,q}^{m,n}&\left(z\biggl|\begin{array}{cccc}
                                                                a_1, & a_2, & \ldots, & a_p \\
                                                                b_1, & b_2, & \ldots, & b_q
                                                              \end{array}
\right)=\frac{1}{2\pi i}\int\limits_{C}\frac{\prod_{j=1}^m\Gamma(b_j-s)\prod_{j=1}^n\Gamma(1-a_j+s)}{\prod_{j=m+1}^q\Gamma(1-b_j+s)\prod_{j=n+1}^p\Gamma(a_j-s)}z^sds.
\nonumber
\end{split}
\end{equation}
Here an empty product is interpreted as unity, $0\leq m\leq q$, $0\leq n\leq p$, and the parameters
$\{a_k\}$ ($k=1,\ldots,p$) and $\{b_j\}$ ($j=1,\ldots,m$) are such that no pole of $\Gamma(b_j-s)$ coincides with any pole of
$\Gamma(1-a_k+s)$. We assume that $z\in\C\setminus\{0\}$. The contour of integration $C$ goes from $-i\infty$ to $+i\infty$ so that all poles
$\Gamma(b_j-s)$, $j=1,\ldots,m,$ lie to the right of the path, and all poles of $\Gamma(1-a_k+s)$, $k=1,\ldots,n,$ lie to the left of the path.
If $p=0$, then $n=0$, and  we write the corresponding Meijer $G$-function as
$
G_{0,q}^{m,0}\left(x\biggl|\begin{array}{cccc}
 b_1, & b_2, & \ldots, & b_q
\end{array}
\right).
$

 Let $X(1)$, $\ldots$, $X(M)$ be independent matrices whose entries are i.i.d standard complex Gaussian variables.
 Assume that $X(j)$ has size $N_j\times N_{j-1}$, where $1\leq j\leq M$. Consider the product matrix
 $$
 Y_M=X(M)X(M-1)\ldots X(1),
 $$
 which is a rectangular matrix of size $N_{M}\times N_0$. Note that $Y_M^*Y_M$
 is a square matrix of size $N_0$. Set
 $$
 \nu_l=N_l-N_0,\;\;l=0,1,\ldots,M,
 $$
 and assume that $\nu_l\geq 0$ for $l=1,\ldots, M$.\footnote{Note that $\nu_0=0$. This fact will be used throughout the paper.}
 Akemann, Ipsen and Kieburg \cite{AkemannIpsenKieburg} proved the following
 \begin{thm}\label{TheoremSingularvalues}
 The squared singular values of $Y_M$ form a determinantal point process on $\R_{>0}$.
 This determinantal point process is a biorthogonal ensemble\footnote{The notion of biorthogonal ensembles was introduced in Borodin \cite{Borodin}} with joint density function given by
 $$
 P^{(M)}(x_1,\ldots,x_{N_0})=\frac{1}{Z_{N_0}}\prod\limits_{1\leq j<k\leq N_0}(x_k-x_j)\det\left(w_{k-1}^{(M)}(x_j)\right)_{j,k=1}^{N_0},
 $$
 where $x_k>0$, $k=1,\ldots,N_0$, are the squared singular values of $Y_M$, $Z_{N_0}$ is a normalization constant,  and the functions $w_k^{(M)}(x)$ can be expressed in terms of Meijer's G-functions,
 $$
 w_k^{(M)}(x)=G_{0,M}^{M,0}\left(x\biggl|\begin{array}{ccccc}
 \nu_M, & \nu_{M-1}, & \ldots, &\nu_2, & \nu_1+k
\end{array}
\right).
 $$
 \end{thm}
  As $M=1$ we obtain the determinantal point process on $\R_{>0}$ called the classical Laguerre ensemble.

Kuijlaars and Zhang \cite{Kuijlaars} found the scaling limit at the origin of the relevant correlation kernel which generalizes the classical Bessel kernel.  According to Kuijlaars and Zhang \cite{Kuijlaars}, the limiting  kernel  can be written explicitly in terms of  Meijer's G-functions.
Namely, the following Theorem holds true
\begin{thm}
Let $K_{N_0}^{(M)}(x,y)$ be the correlation kernel of the determinantal point process defined by the joint density $P^{(M)}(x_1,\ldots,x_{N_0})$ (see
Theorem \ref{TheoremSingularvalues}). Then we have
$$
\underset{N_0\rightarrow\infty}{\lim}\frac{1}{N_0}K_{N_0}^{(M)}\left(\frac{x}{N_0},\frac{y}{N_0}\right)=K_M(x,y),
$$
where $K_M(x,y)$ is given by
\begin{equation}\label{KuijlaarsZhangKernel}
\begin{split}
&K_M(x,y)\\
&=\frac{\mathcal{B}\left(G_{0,M+1}^{1,0}\left(x\biggl|\begin{array}{cccc}
-\nu_0, & -\nu_1, & \ldots, & -\nu_M
\end{array}
\right),G_{0,M+1}^{M,0}\left(y\biggl|\begin{array}{cccc}
\nu_1,  & \ldots, & \nu_M, & \nu_0
\end{array}
\right)\right)}{x-y}.
\end{split}
\end{equation}
In the formula above  $\mathcal{B}(.,.)$ is a bilinear operator defined by
$$
\mathcal{B}\left(f(x),g(y)\right)=(-1)^{M+1}\sum\limits_{j=0}^M(-1)^j\left(x\frac{d}{dx}\right)^jf(x)\left(\sum\limits_{i=0}^{M-j}\alpha_{i+j}\left(y\frac{d}{dy}\right)^ig(y)\right).
$$
The constants $\alpha_i$  are determined by
\begin{equation}\label{Coefficientsa_i}
\prod\limits_{i=1}^M\left(x-\nu_i\right)=\sum\limits_{i=0}^M\alpha_ix^i,
\end{equation}
i.e. the constants  $\alpha_i$  are expressible in terms of the elementary symmetric polynomials,
$$
\alpha_i=(-1)^ie_{M-i}(\nu_1,\ldots,\nu_M).
$$
\end{thm}
Using the fact that $G_{0,2}^{1,0}\left(x\biggl|\begin{array}{cc}
\nu_1,  &  \nu_0
\end{array}
\right) $ can be written in terms of the Bessel functions it is not
hard to check (see  Kuijlaars and Zhang \cite{Kuijlaars}, Section
5.3) that if $M=1$ and $\nu_1=\nu$, one obtains a kernel equivalent
to the classical Bessel kernel (equation (\ref{BesselKernel})).
 It was observed in Kuijlaars and Zhang \cite{Kuijlaars} that for $M=2$ kernel (\ref{KuijlaarsZhangKernel}) coincides with the scaling limit found by
Bertola, Gekhtman and Szmigielski \cite{Bertola} in the Cauchy-Laguerre two-matrix model. Moreover, Forrester \cite{Forrester} has proved that
the kernel $K_M(x,y)$  also arises (at the same scaling limit) in the problem on eigenvalue statistics for complex $N\times N$ Wishart matrices $W^{*}_{r,s}W_{r,s}$,
where $W_{r,s}$ is the product of $r$ complex Gaussian matrices, and the inverse of $s$ complex Gaussian matrices.

Kuijlaars and Zhang \cite{Kuijlaars} noted that the limiting kernel
(\ref{KuijlaarsZhangKernel}) can be represented  in the form
\begin{equation}\label{IntegrableForm}
K(x,y)=\frac{\sum\limits_{i=1}^kF_i(x)G_i(y)}{x-y},\;\;\;\mbox{where}\;\;\sum\limits_{i=1}^kF_i(x)G_i(x)=0.
\end{equation}
Such kernels (and operators defined by such kernels) are called integrable, see Its, Izergin, Korepin and Slavnov
\cite{Its}, Deift \cite{DeiftIntegrableOperators}.
Many questions in the theory of classical ensembles of random matrices can be reduced
to the evaluation of Fredholm determinants $\det(I-\lambda K)$, where $K$ is an operator which has a kernel
of form (\ref{IntegrableForm}).  For a general theory of such Fredholm determinants
see, for example, Its \cite{Its1}, Its and Harnad \cite{ItsHarnad}, Deift \cite{DeiftIntegrableOperators}, Tracy and Widom \cite{TracyWidom,TracyWidomSystems} and
the references therein.

To summarize, the results obtained in Akemann, Ipsen and Kieburg \cite{AkemannIpsenKieburg}, Kuijlaars and Zhang \cite{Kuijlaars} imply that
after rescaling at the hard edge (origin) the statistics of singular values is described by Fredholm determinants of integrable operators
(in the sense of Its, Isergin, Korepin and Slavnov \cite{Its}).

\section{Main results}\label{SectionMainResults}
\subsection{Notation}
Let $f(x)$ and $g(y)$ be two functions defined in terms of Meijer's
$G$-functions as follows
\begin{equation}\label{f}
f(x)=G_{0,M+1}^{1,0}\left(x\biggl|\begin{array}{cccc}
 -\nu_0, & -\nu_1, & \ldots, & -\nu_M
\end{array}
\right),
\end{equation}
and
\begin{equation}\label{g}
g(y)=G_{0,M+1}^{M,0}\left(y\biggl|\begin{array}{cccc}
 \nu_1, & \ldots, & \nu_M, & \nu_0
\end{array}
\right).
\end{equation}
The functions $f(x)$ and $g(y)$
can be expressed in terms of contour integrals as
\begin{equation}
f(x)=\frac{1}{2\pi
i}\int\limits_{c-i\infty}^{c+i\infty}\frac{\Gamma(-s)x^s}{\prod_{j=1}^M\Gamma(1+\nu_j+s)}ds,
\nonumber
\end{equation}
and
\begin{equation}
g(y)=\frac{1}{2\pi i}\int\limits_{c-i\infty}^{c+i\infty}\frac{\prod_{j=1}^M\Gamma(\nu_j-s)}{\Gamma(1+s)}y^sds,
\nonumber
\end{equation}
where $-1<c<0$. (We have used the fact that $\nu_0=0$).
\begin{prop}
The functions  $f(x)$ and $g(y)$ satisfy the following differential equations
\begin{equation}\label{DiffEquationForf(x)}
\prod\limits_{j=0}^M(x\frac{d}{dx}+\nu_j)f(x)=-xf(x),
\end{equation}
and
\begin{equation}\label{DiffEquationForg(y)}
\prod\limits_{j=0}^M(y\frac{d}{dy}-\nu_j)g(y)=(-1)^Myg(y).
\end{equation}
\end{prop}
\begin{proof}
Use the fact that the Meijer G-function $G_{p,q}^{m,n}\left(z\biggl|\begin{array}{cccc}
                                                                a_1, & a_2, & \ldots, & a_p \\
                                                                b_1, & b_2, & \ldots, & b_q
                                                              \end{array}
\right)$ satisfies the following differential equation
\begin{equation}
\begin{split}
\biggl[(-1)^{p-m-n}z&\prod\limits_{j=1}^p\left(z\frac{d}{dz}-a_j+1\right)\\
&-\prod\limits_{j=1}^q\left(z\frac{d}{dz}-b_j\right)\biggr]G_{p,q}^{m,n}\left(z\biggl|\begin{array}{cccc}
                                                                a_1, & a_2, & \ldots, & a_p \\
                                                                b_1, & b_2, & \ldots, & b_q
                                                              \end{array}
\right)=0,
\end{split}
\end{equation}
see Ref. \cite{Olver}, formula (16.21.1).
\end{proof}
For $0\leq j\leq M$ define
\begin{equation}\label{PHIPSI}
\phi_j(x)=(-1)^{M-j+1}\left(x\frac{d}{dx}\right)^jf(x),\;\;\psi_j(y)=\sum\limits_{i=0}^{M-j}\alpha_{i+j}\left(y\frac{d}{dy}\right)^{i}g(y),
\end{equation}
where the coefficients $\alpha_0,\ldots,\alpha_M$ are defined by equation (\ref{Coefficientsa_i}).
Using these functions, the correlation kernel $K_M(x,y)$  (defined by equation (\ref{KuijlaarsZhangKernel})) can be written as
\begin{equation}\label{The Generalized Bessel Kernel}
K_M(x,y)=\frac{\sum\limits_{j=0}^M\phi_j(x)\psi_j(y)}{x-y},
\end{equation}
where the indeterminacy arising for $x=y$ is resolved via the L'Hospital rule.
As it is noted in Kuijlaars and Zhang \cite{Kuijlaars}, Section 5.2, the kernel $K_M(x,y)$ has the continuity property. This means that the condition
\begin{equation}\label{ContinuityOfK_M}
\sum\limits_{j=0}^M\phi_j(x)\psi_j(x)=0
\end{equation}
holds true. In what follows we will refer to the kernel $K_M(x,y)$ defined by equation (\ref{KuijlaarsZhangKernel}) or, equivalently, by equation (\ref{The Generalized Bessel Kernel}) as  the \textit{generalized Bessel kernel}.
\begin{prop} The correlation kernel $K_M(x,y)$ can be written as
\begin{equation}\label{EquationKernelAsIntegral}
K_M(x,y)=\int\limits_{0}^1f(xt)g(yt)dt.
\end{equation}
\end{prop}
\begin{proof}
See Kuijlaars and Zhang \cite{Kuijlaars}, Section 5.2.
\end{proof}
Let $0\leq a_1<a_2<\ldots<a_{2m}$, and denote by $J$ the union of intervals of the form $(a_{2j-1},a_{2j})$. Thus $J=\bigcup\limits_{j=1}^{m}\left(a_{2j-1},a_{2j}\right)$.
Denote by $K_M$ the operator acting in $L^2(0,+\infty)$ with the kernel $K_M(x,y)\chi_J(y)$, where $\chi_j$ denotes the characteristic function of the interval $J$. The Fredholm determinant $\det(1-K_M)$ gives the probability that no particles of the limiting determinantal point process
lie in $J$. Thus we assume without proof that $\det(1-K_M)$ exists, and that it is not equal to zero. This implies that $\lambda=1$
is not an eigenvalue of $K_M$, and
that $(1-K_M)$ is invertible. Note that there are different ways to choose the operator $K_M$ such that the Fredholm determinant $\det\left(1-K_M\right)$
will give the gap probability. Our specific choice enables us to keep calculations simple.

Define the operators $R_M$ and $\rho_M$ by
\begin{equation}\label{R}
R_M=(1-K_M)^{-1}K_M=-1+(1-K_M)^{-1},\;\; \rho_M=(1-K_M)^{-1}.
\end{equation}
In addition, denote by $K'_M$ the operator with the kernel $K_M(y,x)\chi_J(y)$, and by $K_M^T$ the operator with the kernel $K_M(y,x)\chi_J(x)$.
In the same way we define the operators $R_M'$ and $\rho_M'$, i.e.
\begin{equation}\label{R'}
R_M'=(1-K_M')^{-1}K_M'=-1+(1-K_M')^{-1},\;\; \rho_M'=(1-K_M')^{-1},
\end{equation}
and the operators $R_M^{T}$ and $\rho_M^{T}$,
\begin{equation}
R_M^T=(1-K_M^T)^{-1}K_M^T=-1+(1-K_M^T)^{-1},\;\; \rho_M^T=(1-K_M^T)^{-1}.
\end{equation}
\subsection{Partial differential equations}\label{SectionPartial differential equations}
In this Section we construct a system of nonlinear partial differential equations associated with the Fredholm determinant $\det(1-K_M)$,
where the endpoints of the intervals, $a_1$, $\ldots$, $a_{2m}$ serve as independent variables. This system of differential equations will be interpreted as a Hamiltonian system in Section  \ref{SectionHamiltonianEquations}.
For $0\leq j\leq M$ and $1\leq l\leq m$ define the following quantities
$$
x_j^{(2l)}:=\sqrt{-1}\left(1-K_M\right)^{-1}\phi_j(a_{2l}),
$$
$$
y_j^{(2l)}:=\sqrt{-1}\left(1-K_M'\right)^{-1}\psi_j(a_{2l}),
$$
$$
x_j^{(2l-1)}:=\left(1-K_M\right)^{-1}\phi_j(a_{2l-1}),
$$
$$
y_j^{(2l-1)}:=\left(1-K_M'\right)^{-1}\psi_j(a_{2l-1}),
$$
\begin{equation}
\begin{split}
\xi_j:=(-1)^M\sum\limits_{k=1}^{2m}\int\limits_{a_{2k-1}}^{a_{2k}}\phi_0(x)\left(1-K_M'\right)^{-1}\psi_j(x)dx+(-1)^{M+1-j}e_{M+1-j}(\nu_0,\ldots,\nu_M),
\end{split}
\nonumber
\end{equation}
(where $e_{M+1-j}(\nu_0,\ldots,\nu_M)$ is the $({M+1-j})$th elementary symmetric polynomial  in the variables $\nu_0,\ldots,\nu_M$),
$$
\eta_j:=(-1)^M\sum\limits_{k=1}^{2m}\int\limits_{a_{2k-1}}^{a_{2k}}\phi_j(x)\left(1-K_M'\right)^{-1}\psi_M(x)dx.
$$
These quantities are  functions of parameters $a_1$, $\ldots$, $a_{2m}$. In Section \ref{SectionHamiltonianEquations} these quantities
will be interpreted as canonical coordinates of a Hamiltonian system.
\begin{prop}\label{Proposition Partial Differential Equations for x y xi eta}
The functions $x_j^{(k)}$, $y_j^{(k)}$, $\xi_j$ and $\eta_j$  satisfy  the following system
of partial differential equations\\
\textbf{(a)} For $0\leq j\leq M$, and $1\leq k\neq l\leq 2m$
 \begin{equation}\label{DynamicalEquation(a)}
 \frac{\partial x_j^{(l)}}{\partial a_k}=-\frac{x_j^{(k)}}{a_l-a_k}\sum\limits_{i=0}^Mx_i^{(l)}y_i^{(k)}.
 \end{equation}
 \textbf{(b)} For $0\leq j\leq M$, and $1\leq k\neq l\leq 2m$
 \begin{equation}\label{JMMS2}
 \frac{\partial y_j^{(l)}}{\partial a_k}=-\frac{y_j^{(k)}}{a_k-a_l}\sum\limits_{i=0}^Mx_i^{(k)}y_i^{(l)}.
 \end{equation}
 \textbf{(c)} For $0\leq j\leq M-1$, and $1\leq  l\leq 2m$
 \begin{equation}\label{JMMS3}
 a_l\frac{\partial x_j^{(l)}}{\partial a_l}=-\eta_jx_0^{(l)}-x_{j+1}^{(l)}
 +\underset{k\neq l}{\sum\limits_{k=1}^{2m}}x_j^{(k)}\frac{a_k}{a_l-a_k}\sum\limits_{i=0}^Mx_i^{(l)}y_i^{(k)},
 \end{equation}
 and for $j=M$ and $1\leq l\leq 2m$
 \begin{equation}\label{JMMS3'}
 \begin{split}
 a_l\frac{\partial x_M^{(l)}}{\partial a_l}&=-\eta_Mx_0^{(l)}+(-1)^{M+1}a_lx_0^{(l)}\\
 &+\sum\limits_{i=0}^M\xi_ix_i^{(l)}
 +\underset{k\neq l}{\sum\limits_{k=1}^{2m}}x_M^{(k)}\frac{a_k}{a_l-a_k}\sum\limits_{i=0}^Mx_i^{(l)}y_i^{(k)}.
 \end{split}
 \end{equation}
 \textbf{(d)} For $1\leq j\leq M$, and $1\leq  l\leq 2m$
 \begin{equation}\label{JMMS4}
 a_l\frac{\partial y_j^{(l)}}{\partial a_l}=-\xi_jy_M^{(l)}+y_{j-1}^{(l)}
 +\underset{k\neq l}{\sum\limits_{k=1}^{2m}}y_j^{(k)}\frac{a_k}{a_k-a_l}\sum\limits_{i=0}^Mx_i^{(k)}y_i^{(l)},
 \end{equation}
 and for $j=0$ and $1\leq l\leq 2m$
 \begin{equation}\label{JMMS4'}
 \begin{split}
 a_l\frac{\partial y_0^{(l)}}{\partial a_l}&=-\xi_0y_M^{(l)}+(-1)^{M}a_ly_M^{(l)}\\
 &+\sum\limits_{i=0}^M\eta_iy_i^{(l)}
 +\underset{k\neq l}{\sum\limits_{k=1}^{2m}}y_0^{(k)}\frac{a_k}{a_k-a_l}\sum\limits_{i=0}^Mx_i^{(k)}y_i^{(l)}.
 \end{split}
 \end{equation}
 \textbf{(e)} For $0\leq i,j\leq M$, and $1\leq  l\leq 2m$
 \begin{equation}\label{DynamicalEquation(e)}
 \frac{\partial}{\partial a_l}\xi_j=(-1)^{M+1}x_0^{(l)}y_j^{(l)}.
 \end{equation}
 \textbf{(f)} For $0\leq i,j\leq M$, and $1\leq  l\leq 2m$
 \begin{equation}\label{DynamicalEquation(f)}
 \frac{\partial}{\partial a_l}\eta_j=(-1)^{M+1}x_j^{(l)}y_M^{(l)}.
 \end{equation}
\end{prop}
 We will refer to equations (\ref{DynamicalEquation(a)})-(\ref{DynamicalEquation(f)}) as  the system of dynamical equations associated with the generalized Bessel kernel
 $K_M(x,y)$ defined by equation (\ref{The Generalized Bessel Kernel}).\\
 \textbf{Remarks.}
 \\
a) The quantity $(1-K_M)^{-1}\phi_j(a_k)$ means $\underset{x\rightarrow a_k}{\underset{x\in J}{\lim}}\left[(1-K_M)^{-1}\phi_j\right](x)$. The same meaning has the quantity $(1-K_M')\psi_j(a_k)$.\\
b) For the sine kernel partial differential equations similar to that of Proposition \ref{Proposition Partial Differential Equations for x y xi eta}
were first derived by  Jimbo, Miwa, M$\hat{\mbox{o}}$ri and Sato \cite{Jimbo}. These equations are called the JMMS equations in the random matrix literature. Tracy and Widom \cite{TracyWidomIntroduction}-\cite{TracyWidom} derived analogues of the JMMS equations for correlation kernels of the form
$$
\frac{A(x)B(y)-A(y)B(x)}{x-y}.
$$
Equations (\ref{DynamicalEquation(a)})-(\ref{DynamicalEquation(f)}) of Proposition \ref{Proposition Partial Differential Equations for x y xi eta}
are partial differential equations for a correlation kernel of even more general form (\ref{IntegrableForm}).\\
c) Assume that $M=1$ and $\nu_1=0$ (recall that $\nu_0=0$). In this case the correlation kernel has an especially simple form. Namely,
$$
f(x)=g(x)=G_{0,2}^{1,0}\left(x\biggl|\begin{array}{cc}
 0, & 0
\end{array}
\right)=J_{0}(2\sqrt{x}).
$$
The functions $\phi_0(x)$, $\phi_1(x)$, $\psi_0(y)$, and $\psi_1(y)$ are defined by equations
\begin{equation}
\phi_0(x)=f(x), \;\; \phi_1(x)=-x\frac{d}{dx}f(x),
\end{equation}
and
\begin{equation}
\psi_0(y)=y\frac{d}{dy}f(y),\;\; \psi_1(y)=f(y).
\end{equation}
In particular, we have
$$
\phi_0=\psi_1,\;\; \phi_1=-\psi_0,
$$
and the kernel of $K_{M=1}$ can be written as
$$
K_{M=1}(x,y)\chi_J(y),\;\; K_{M=1}(x,y)=\frac{A(x)B(y)-A(y)B(x)}{x-y},
$$
where the functions $A(x)$, $B(y)$ are defined by
$$
A(x)=f(x),\;\; B(y)=y\frac{d}{dy}f(y).
$$
In this case
the partial differential equations (\ref{DynamicalEquation(a)})-(\ref{DynamicalEquation(f)}) turn into partial differential equations for the Bessel kernel derived in Tracy and Widom
 \cite{Tracy}, see equations (1.9)-(1.14) of Tracy and Widom
 \cite{Tracy}.\\
d) The prefactors $\sqrt{-1}$ in the definition of $x_j^{(2l)}$ and $y_j^{(2l)}$
are needed for unified differential equations for $\xi_j$ and $\eta_j$, see equations (\ref{DynamicalEquation(e)}) and (\ref{DynamicalEquation(f)}). Namely, the prefactors $\sqrt{-1}$ enables one to avoid a case discussion for odd and even  $l$. It is unimportant which root, $\pm i$, is taken for $\sqrt{-1}$.
 \subsection{Hamiltonian structure of dynamical equations associated with the generalized Bessel kernel}\label{SectionHamiltonianEquations}
 Here we claim that the  system  of dynamical equations associated with the generalized Bessel kernel
 $K_M(x,y)$ can be understood as a  Hamiltonian system with $(2m+1)(M+1)$ canonical conjugate pairs.

 For $1\leq l\leq 2m$ set
 \begin{equation}\label{HlDefinition}
 H_l:=-a_l\frac{\partial}{\partial a_l}\log\left(\det\left(1-K_M\right)\right).
 \end{equation}
Introduce the Poisson brackets by
 \begin{equation}\label{Symplectic Structure}
 \left\{x_j^{(l)},y_i^{(k)}\right\}=\frac{1}{a_l}\delta_{l,k}\delta_{i,j},\;\; \left\{\xi_j,\eta_i\right\}=(-1)^M\delta_{i,j}.
 \end{equation}
 In other words, for two functions $F$ and $G$ of dynamical variables $\left(x_j^{(k)},\xi_j;y_j^{(k)},\eta_j\right)$
 (where $1\leq k\leq 2m$, and $0\leq j\leq M$) we  define the Poisson brackets by the formula
 \begin{equation}\label{Poisson Brackets}
 \left\{F,G\right\}:=
 \sum\limits_{k=1}^{2m}\frac{1}{a_k}\sum\limits_{j=0}^M\left(\frac{\partial F}{\partial x_j^{(k)}}\frac{\partial G}{\partial y_j^{(k)}}-\frac{\partial F}{\partial y_j^{(k)}}\frac{\partial G}{\partial x_j^{(k)}}\right)
 +(-1)^M\sum\limits_{j=0}^M\left(\frac{\partial F}{\partial\xi_j}\frac{\partial G}{\partial\eta_j}-\frac{\partial F}{\partial\eta_j}\frac{\partial G}{\partial\xi_j}\right).
 \end{equation}
 \begin{prop}\label{PropositionHamiltonianEquations}
 The system of dynamical equations (\ref{DynamicalEquation(a)})-(\ref{DynamicalEquation(f)}) associated with the generalized Bessel kernel  $K_M(x,y)$ (defined by equation (\ref{The Generalized Bessel Kernel})) can be written as
\begin{equation}\label{xdot1,ydot1}
\frac{\partial x_j^{(k)}}{\partial a_l}=\left\{x_j^{(k)},H_l\right\},\;\;\frac{\partial y_j^{(k)}}{\partial a_l}=\left\{y_j^{(k)},H_l\right\},
\end{equation}
and
\begin{equation}\label{zetadot1,etadot1}
\frac{\partial \xi_j}{\partial a_l}=\left\{\xi_j,H_l\right\},\;\;\frac{\partial\eta_j}{\partial a_l}=\left\{\eta_j,H_l\right\},
\end{equation}
where $0\leq j\leq M$, and $1\leq k,l\leq 2m$. The Hamiltonians are given explicitly by
\begin{equation}\label{Hl}
 \begin{split}
 &H_l=-\left(\sum\limits_{j=0}^M\eta_jy_j^{(l)}\right)x_0^{(l)}
 +(-1)^{M+1}a_lx_0^{(l)}y_M^{(l)}\\
 &-\sum\limits_{j=0}^{M-1}x_{j+1}^{(l)}y_j^{(l)}+y_M^{(l)}\sum\limits_{k=0}^M\xi_kx_k^{(l)}\\
 &+\underset{k\neq l}{\sum\limits_{k=1}^{2m}}\frac{a_k}{a_l-a_k}\sum\limits_{i,j=0}^Mx_i^{(k)}x_j^{(l)}y_i^{(l)}y_j^{(k)}.
 \end{split}
 \end{equation}
\end{prop}
\begin{prop}\label{PropositionInvolution}
The Hamiltonians $H_l$ are in involution. This means that for all $1\leq l,\rho\leq 2m$ the following condition is satisfied
\begin{equation}\label{Involution}
\left\{H_l,H_{\rho}\right\}=0,
\end{equation}
where the symplectic structure is defined by  equation (\ref{Symplectic Structure}).
\end{prop}
\textbf{Remarks.}\\
a) A similar Hamiltonian interpretation is known for other systems of partial differential equations associated with correlation kernels of Random Matrix Theory, see
Tracy and Widom \cite{TracyWidomIntroduction}-\cite{TracyWidom}.\\
b) It is shown in Tracy and Widom \cite{TracyWidomIntroduction} (in
the context of the sine kernel) that the fact that the Hamiltonians
are in involution implies the complete integrability of partial
differential equations (\ref{xdot1,ydot1}), (\ref{zetadot1,etadot1})
(or, equivalently, of equations
(\ref{DynamicalEquation(a)})-(\ref{DynamicalEquation(f)})) in the
sense of Frobenius.\\
c) For each $1\leq l\leq m$ set
$$
h_l=\frac{\partial}{\partial a_l}\log\left(\det\left(1-K_M\right)\right).
$$
Introduce the $1$-form
$$
w(a_1,\ldots,a_{2m})=\sum\limits_{l=1}^{2m}h_lda_l.
$$
Using  the fact that the Hamiltonians $H_l$ are in involution (see Proposition \ref{PropositionInvolution}) we obtain that
$w(a_1,\ldots,a_{2m})$ is locally an exact differential
$$
w(a_1,\ldots,a_{2m})=d\log\tau.
$$
This equation defines (up to a multiple constant) the tau-function associated with the dynamical equations (\ref{xdot1,ydot1}) and (\ref{zetadot1,etadot1}).
We conclude that this tau-function  evaluated on a solution of the dynamical equations (\ref{xdot1,ydot1}) and (\ref{zetadot1,etadot1}) is given by the Fredholm determinant
\begin{equation}
\tau(a_1,\ldots,a_{2m})=\det\left(1-K_M\right).
\end{equation}
\subsection{The isomonodromic system associated with the generalized Bessel kernel}\label{SubsectionIsomonodromicSystem}
The aim of this Section is to construct isomonodromic  deformation equations associated with the generalized Bessel kernel, and with
the Fredholm determinant $\det(1-K_M)$. For this purpose set
\begin{equation}\label{E}
E=(-1)^{M+1}\left(\begin{array}{cccc}
                    0 & 0 & \ldots & 0 \\
                    0 & 0 & \ldots & 0 \\
                    \vdots &  &  &  \\
                    1 & 0 & \ldots & 0
                  \end{array}
\right),
\end{equation}
and
\begin{equation}\label{A,C}
A^{(l)}=\left(\begin{array}{cccc}
                x_0^{(l)}y_0^{(l)} &  x_0^{(l)}y_1^{(l)} & \ldots &  x_0^{(l)}y_M^{(l)} \\
                 x_1^{(l)}y_0^{(l)} &  x_1^{(l)}y_1^{(l)} & \ldots &  x_1^{(l)}y_M^{(l)} \\
                \vdots &  &  &  \\
                x_M^{(l)}y_0^{(l)} & x_M^{(l)}y_1^{(l)} & \ldots & x_M^{(l)}y_M^{(l)}
              \end{array}
\right),\;\;
C=\left(\begin{array}{ccccc}
          -\eta_0 & -1 & 0 & \ldots & 0 \\
          -\eta_1 & 0 & -1 & \ldots & 0 \\
          \vdots &  &  &  &  \\
          -\eta_{M-1} & 0 & 0 & \ldots & -1 \\
          \xi_0-\eta_M & \xi_1 & \xi_2 & \ldots & \xi_M
        \end{array}
\right),
\end{equation}
where $1\leq l\leq 2m$.
\begin{prop}\label{PropositionSchlesinger}
 The system of dynamical equations (\ref{DynamicalEquation(a)})-(\ref{DynamicalEquation(f)}) associated with the generalized Bessel kernel
 $K_M(x,y)$ defined by equation (\ref{The Generalized Bessel Kernel}) implies
\begin{equation}\label{Schlesinger1}
\frac{\partial}{\partial a_k}A^{(l)}=\frac{\left[A^{(l)}, A^{(k)}\right]}{a_l-a_k},\;\; 1\leq l\neq k\leq 2m,
\end{equation}
\begin{equation}
a_l\frac{\partial A^{(l)}}{\partial a_l}=\left[C+a_lE,A^{(l)}\right]+
\underset{k\neq l}{\sum\limits_{k=1}^{2m}}\frac{a_k}{a_l-a_k}\left[A^{(k)},A^{(l)}\right],\;\; 1\leq l\leq 2m,
\end{equation}
\begin{equation}\label{Schlesinger3}
\frac{\partial C}{\partial a_l}=\left[E, A^{(l)}\right],\;\; 1\leq l\leq 2m.
\end{equation}
Moreover, we have
$$
H_l=\Tr\left(CA^{(l)}\right)+a_l\Tr\left(EA^{(l)}\right)+\underset{k\neq l}{\sum\limits_{k=1}^{2m}}\frac{a_k}{a_l-a_k}\Tr\left(A^{(k)}A^{(l)}\right),
$$
where $1\leq l\leq 2m$.
\end{prop}
Consider the following linear system of ordinary differential equations with rational coefficients
\begin{equation}\label{Isomonodromy1}
\frac{d\Psi}{dz}=X(z)\Psi,
\end{equation}
where $X(z)$ is a $(M+1)\times (M+1)$ matrix defined by
\begin{equation}\label{Isomonodromy2}
X(z)=E+\frac{C-\sum\limits_{k=1}^{2m}A^{(k)}}{z}+\sum\limits_{k=1}^{2m}\frac{A^{(k)}}{z-a_k},
\end{equation}
and where $E$, $C$, $A^{(k)}$ are certain $(M+1)\times (M+1)$ matrices which are independent on $z$.
Consider the poles $a_1,\ldots,a_{2m}$ as deformation parameters of the coefficient matrix $X(z)$.
Thus
\begin{equation}\label{Isomonodromy3}
X(z)=X(z;a_1,\ldots,a_{2m}).
\end{equation}
Suppose that $\Psi(z;a_1,\ldots,a_{2m})$ (in addition to equation (\ref{Isomonodromy1})) also satisfies
a linear system of ordinary differential equations with respect to the parameters $a_1,\ldots, a_{2m}$
\begin{equation}\label{Isomonodromy4}
\frac{\partial\Psi}{\partial a_j}=\Theta_j(z)\Psi,\;\;\;\Theta_j(z)=-\frac{A^{(j)}}{z-a_j},\;\;\; 1\leq j\leq 2m.
\end{equation}
Note that we can write equations (\ref{Isomonodromy4}) as
\begin{equation}\label{Isomonodromy5}
d\Psi=\Theta(z)\Psi,\;\;\; \Theta(z)=\sum\limits_{j=1}^{2m}\Theta_j(z)da_j.
\end{equation}
The compatibility condition of equations (\ref{Isomonodromy1}) and (\ref{Isomonodromy5}) (i.e. $d\frac{d\Psi}{dz}=\frac{d}{dz}d\Psi$) gives
\begin{equation}\label{IsomonodromyDeformationEquation}
dX=\frac{\partial\Theta}{\partial z}+\left[\Theta,X\right].
\end{equation}
This equation is called the \textit{isomonodromy deformation equation}
in the  Jimbo, Miwa, M$\hat{\mbox{o}}$ri, Ueno and Sato theory \cite{Jimbo}-\cite{Jimbo2} of isomonodromy deformations, and it gives a condition for
deformation (\ref{Isomonodromy3}) to be isomonodromic. For the modern presentation of the theory
of isomonodromy deformations we refer the reader to the book by Fokas, Its, Kapaev and Novokshenov \cite{FokasItsKapaevNovokshenovBook}, Chapter 4.

Equation (\ref{IsomonodromyDeformationEquation}) leads to
\begin{equation}\label{IsomonodromyDeformationEquation1}
\frac{\partial X}{\partial a_j}=\frac{\partial \Theta_j}{\partial z}+\left[\Theta_j,X\right],\;\; j=1,\ldots,2m.
\end{equation}
Equating the coefficients in equation (\ref{IsomonodromyDeformationEquation1}) before the factors $\frac{1}{z}$ and $\frac{1}{z-a_k}$, $k=1,\ldots, 2m$ (which are independent because the variable $z$ is arbitrary)
we obtain equations (\ref{Schlesinger1})-(\ref{Schlesinger3}) for the matrix coefficients $A^{(k)}$, $k=1,\ldots,2m$, and $C$. Thus the following statement holds true.
\begin{prop}\label{PropositionIsomonodromyFormulation}
Equations (\ref{Schlesinger1})-(\ref{Schlesinger3}) follow from the
isomonodromy deformation equation
(\ref{IsomonodromyDeformationEquation}), with matrices $X(z)$ and
$\Theta(z)$ defined by equations (\ref{Isomonodromy2}),
(\ref{Isomonodromy4}), and (\ref{Isomonodromy5}).
\end{prop}
Now, let us recall the notion of the (Frobenius) complete integrability. Consider the following system of first-order partial differential equations
\begin{equation}\label{SystemofFirstOrderPartialDifferentialEquations}
\frac{\partial y_j(b_1,\ldots,b_n)}{\partial b_l}=\varphi_{j,l}\left(y_1(b_1,\ldots,b_n),\ldots,y_N(b_1,\ldots,b_n),b_1,\ldots,b_n\right),
\end{equation}
where $1\leq j\leq N$, and $1\leq l\leq n$. Here $(b_1,\ldots,b_n)\in\C^{n}$ are independent variables, $y_1,\ldots,y_N$ are unknown functions of
$(b_1,\ldots,b_n)$, and $\varphi_{j,l}$ are given holomorphic functions defined in a domain $\mathcal{D}\subset \C^{N}\times\C^n$.
\begin{defn}
The system of first-order partial differential equations (\ref{SystemofFirstOrderPartialDifferentialEquations}) is called completely integrable (in the sense of Frobenius)
if for any $\left(z_1,\ldots,z_N,\zeta_1,\ldots,\zeta_n\right)\in\mathcal{D}$ there exists a solution of (\ref{SystemofFirstOrderPartialDifferentialEquations}) such that
$$
y_j(\zeta_1,\ldots,\zeta_n)=z_j,\;\;\; 1\leq j\leq N.
$$
\end{defn}

It is known that isomonodromy  deformation equations associated with systems of linear ordinary differential
 equations with rational coefficients are completely integrable in the sense of Frobenius. Thus we conclude that
equations (\ref{Schlesinger1})-(\ref{Schlesinger3}) are completely
integrable
in the sense of Frobenius, see Jimbo, Miwa, and Ueno \cite{Jimbo1}.\\
\textbf{Remarks.}\\
a) The system
\begin{equation}
\left\{
  \begin{array}{ll}
   \frac{d\Psi}{dz}=X(z)\Psi  \\
   d\Psi=\Theta(z)\Psi
  \end{array}
\right.
\end{equation}
(where the matrices $X(z)$ and $\Theta(z)$ defined by equations
(\ref{Isomonodromy2}), (\ref{Isomonodromy4}), (\ref{Isomonodromy5})
and the matrices $E$, $C$, $A^{(l)}$ are defined by equations
(\ref{E}),  (\ref{A,C})) can be understood as a Lax representation
of the isomonodromy deformation equation
(\ref{IsomonodromyDeformationEquation}).\\
b) Equations (\ref{Schlesinger1})-(\ref{Schlesinger3}) are analogues
of the Schlesinger equations  of the theory of isomonodromy
deformations.
\subsection{The special case of $J=(0,s)$}\label{SubsectionSpecialCase}
This case corresponds to $m=1$, $a_1=0$, $a_2=s$.
\begin{prop}\label{PropositionNonlinearOrdinaryDiffEquation}
\textbf{(A)} In the special case  $J=(0,s)$ partial differential equations (\ref{JMMS3})-(\ref{DynamicalEquation(f)})
lead to the following system of non-linear ordinary differential equations
\begin{equation}\label{NonlinearOrdinaryDiffEquation1}
s\frac{dx_j(s)}{ds}=-\eta_j(s)x_0(s)-x_{j+1}(s),\;\; 0\leq j\leq M-1,
\end{equation}
\begin{equation}
s\frac{dx_M(s)}{ds}=-\eta_M(s)x_0(s)+(-1)^{M+1}sx_0(s)+\sum\limits_{i=0}^M\xi_i(s)x_i(s),
\end{equation}
\begin{equation}
s\frac{dy_j(s)}{ds}=-\xi_j(s)y_M(s)+y_{j-1}(s),\;\; 1\leq j\leq M,
\end{equation}
\begin{equation}
s\frac{dy_0(s)}{ds}=-\xi_0(s)y_M(s)+(-1)^Msy_M(s)+\sum\limits_{i=0}^M\eta_i(s)y_i(s),
\end{equation}
\begin{equation}
\frac{d\xi_j(s)}{ds}=(-1)^{M+1}x_0(s)y_j(s),\;\; 0\leq j\leq M,
\end{equation}
\begin{equation}\label{NonlinearOrdinaryDiffEquation6}
\frac{d\eta_j(s)}{ds}=(-1)^{M+1}x_j(s)y_M(s),\;\;0\leq j\leq M.
\end{equation}
\textbf{(B)}
Set
$$
F_M(s)=\det\left(1-K_M\right),
$$
where $K_M$ is the operator with the kernel
$
K_M(x,y)\chi_{(0,s)}(y)
$
acting on $L^2\left((0,+\infty)\right)$. We have
$$
F_M(s)=\exp\left((-1)^M\int\limits_0^s\log\left(\frac{s}{t}\right)x_0(t)y_M(t)dt\right),
$$
where $x_0(t)$, $y_M(t)$ are components of the solution $(x_j(t),y_j(t),\xi_j(t),\eta_j(t))$ of non-linear ordinary differential equations (\ref{NonlinearOrdinaryDiffEquation1})-(\ref{NonlinearOrdinaryDiffEquation6})
with the initial conditions
\begin{equation}
\begin{split}
&x_j(0)=\sqrt{-1}\phi_j(0),\;\;y_j(0)=\sqrt{-1}\psi_j(0),\;\;\eta_j(0)=0,\\
&\xi_j(0)=(-1)^{M+1-j}e_{M+1-j}(\nu_0,\ldots,\nu_M),\;\; 0\leq j\leq M.
\end{split}
\nonumber
\end{equation}
\end{prop}
\textbf{Remarks.}\\
a) Note that in the case of a single interval $(0,s)$  equations (\ref{Schlesinger1})- (\ref{Schlesinger3}) take an especially simple form. Namely, we have
\begin{equation}
s\frac{d}{ds}A=\left[C+sE,A\right],\; \frac{d}{ds}C=[E,A].
\end{equation}
The Hamiltonian is given by
\begin{equation}
H=\mbox{Tr}\left(CA\right)+s\mbox{Tr}\left(EA\right).
\end{equation}
b) The Fredholm determinant $F_M(s)$ defined above gives the probability that no particles of the determinantal point process
defined by the generalized Bessel kernel (\ref{The Generalized Bessel Kernel}) lie in the interval $(0,s)$.\\
c) In the case $M=1$ the functions $x_0(t)$ and $y_M(t)$ coincide with each other, and we obtain the same result as in Tracy and Widom \cite{Tracy} (equation (1.19)).
In this case equations (\ref{NonlinearOrdinaryDiffEquation1})-(\ref{NonlinearOrdinaryDiffEquation6}) lead to a single ordinary differential equation
for $x_0(t)$ (or for $y_M(t)$) which is reducible to a special case of fifth Painlev$\acute{\mbox{e}}$ equation.\\
d) The initial conditions in Proposition \ref{PropositionNonlinearOrdinaryDiffEquation} come from the very definition of $x_j(s)$, $y_j(s)$, $\xi_j(s)$, and $\eta_j(s)$,
see Section  \ref{SectionPartial differential equations}.

\section{Proofs}\label{SectionProofs}
Our first task is to prove Proposition \ref{Proposition Partial Differential Equations for x y xi eta}. This can be done by adopting methods
developed in  Tracy and Widom \cite{TracyWidomIntroduction}-\cite{TracyWidomSystems} to the situation where the operators are defined by the generalized Bessel kernel $K_M(x,y)$
(equation(\ref{The Generalized Bessel Kernel})). Note that the functions $\phi_j$, $\psi_j$ in formula (\ref{The Generalized Bessel Kernel}) are given explicitly in terms of Meijer's $G$-functions, see equations (\ref{f}), (\ref{g}), and (\ref{PHIPSI}).
\subsection{The functions $\mathcal{Q}_j(x;a_1,\ldots,a_{2m})$, $\mathcal{P}_j(x;a_1,\ldots,a_{2m})$, and $V_{i,j}\left(a_1,\ldots,a_{2m}\right)$}
For $0\leq j\leq M$ introduce the following functions
\begin{equation}\label{FunctionQjDefinition}
\mathcal{Q}_j(x;a_1,\ldots,a_{2m})=(1-K_M)^{-1}\phi_j(x),
\end{equation}
and
\begin{equation}\label{FunctionPjDefinition}
\mathcal{P}_j(x;a_1,\ldots,a_{2m})=(1-K_M')^{-1}\psi_j(x).
\end{equation}
In addition, for $0\leq i,j\leq M$ set
\begin{equation}
V_{i,j}\left(a_1,\ldots,a_{2m}\right)=\int\limits_{J}\phi_i(x)\mathcal{P}_j(x;a_1,\ldots,a_{2m})dx.
\end{equation}
\begin{prop}\label{PropositionMDQ} The functions $\mathcal{Q}_j(x;a_1,\ldots,a_{2m})$ defined by equation (\ref{FunctionQjDefinition}) satisfy  a system of partial differential equations. Namely,
for $0\leq j\leq M-1$ we have
\begin{equation}\label{PartialForQ1}
\begin{split}
&x\frac{\partial}{\partial x}\mathcal{Q}_j(x;a_1,\ldots,a_{2m})
=(-1)^{M+1}V_{j,M}\left(a_1,\ldots,a_{2m}\right)\mathcal{Q}_0(x;a_1,\ldots,a_{2m})
\\
&-\mathcal{Q}_{j+1}(x;a_1,\ldots,a_{2m})
-\sum\limits_{k=1}^{2m}(-1)^ka_kR_M(x,a_k)\mathcal{Q}_j(a_k;a_1,\ldots,a_{2m}),
\end{split}
\end{equation}
and for $j=M$ we have
\begin{equation}\label{PartialForQ2}
\begin{split}
&x\frac{\partial}{\partial x}\mathcal{Q}_M(x;a_1,\ldots,a_{2m})=(-1)^{M+1}V_{M,M}\left(a_1,\ldots,a_{2m}\right)\mathcal{Q}_0(x;a_1,\ldots,a_{2m})\\
&+(-1)^{M+1}x\mathcal{Q}_{0}(x;a_1,\ldots,a_{2m})+\sum\limits_{k=0}^M(-1)^{M+1-k}e_{M+1-k}(\nu_0,\ldots,\nu_M)\mathcal{Q}_k(x;a_1,\ldots,a_{2m})\\
&+
(-1)^M\sum\limits_{k=0}^MV_{0,k}\left(a_1,\ldots,a_{2m}\right)\mathcal{Q}_k(x;a_1,\ldots,a_{2m})\\
&-\sum\limits_{k=1}^{2m}(-1)^ka_kR_M(x,a_k)\mathcal{Q}_M(a_k;a_1,\ldots,a_{2m}).
\end{split}
\end{equation}
In addition, for $1\leq k\leq 2m$, and for $0\leq j\leq M$ we have
\begin{equation}\label{PartialAFromQj(x)}
\begin{split}
\frac{\partial}{\partial a_k}\mathcal{Q}_j(x;a_1,\ldots,a_{2m})=(-1)^kR_M(x,a_k)\mathcal{Q}_j(a_k;a_1,\ldots,a_{2m}).
\end{split}
\end{equation}
In the formulae just written above $R_M(x,a_k)$ stands for the kernel of $R_M$ at points $x$ and $a_k$.
\end{prop}
\begin{proof}
Let us denote by $D$ the operator of differentiation, and by $M$ the operator of multiplication. Thus we have
$$
D\varphi(x)=\frac{d}{dx}\varphi(x), \;\;\;M\varphi(x)=x\varphi(x).
$$
With this notation we have
\begin{equation}\label{MDQj}
\begin{split}
&MD\mathcal{Q}_j(x)=MD(1-K_M)^{-1}\phi_j(x)\\
&=[MD,(1-K_M)^{-1}]\phi_j(x)+(1-K_M)^{-1}MD\phi_j(x), \;\; 0\leq j\leq M.
\end{split}
\end{equation}
Let us first compute the first term in the right-hand side of the equation just written above.
We use the identity
\begin{equation}\label{SimpleOperatorIdentity}
[MD,(1-K_M)^{-1}]=(1-K_M)^{-1}[MD,K_M](1-K_M)^{-1}.
\end{equation}
Thus we need a formula for the commutator $[MD,K_M]$. To find such a formula observe that
for any operator $L$  with the kernel $L(x,y)$ the following identity holds true
\begin{equation}\label{MDL}
\left[MD,L\right](x,y)=\left((MD)_x+(MD)_y+I\right)L(x,y).
\end{equation}
Using this identity we obtain
\begin{equation}\label{P11}
\begin{split}
&\left[MD,K_M\right](x,y)=\left(\left(x\frac{\partial}{\partial x}+y\frac{\partial}{\partial y}\right)K_M(x,y)\right)\chi_J(y)\\
&+y\left(\frac{\partial}{\partial y}\chi_J(y)\right)K_M(x,y)+K_M(x,y)\chi_J(y).
\end{split}
\end{equation}
Using formula (\ref{EquationKernelAsIntegral}) we find that
\begin{equation}\label{P12}
\begin{split}
\left(x\frac{\partial}{\partial x}+y\frac{\partial}{\partial y}\right)K_M(x,y)=\int\limits_0^1t\frac{\partial}{\partial t}\left(f(tx)g(ty)\right)dt\\
=f(x)g(y)-K_M(x,y)=(-1)^{M+1}\phi_0(x)\psi_M(y)-K_M(x,y).
\end{split}
\end{equation}
In addition, we have
\begin{equation}\label{delta}
\frac{\partial}{\partial y}\chi_J(y)=\sum\limits_{k=1}^{2m}(-1)^{k-1}\delta(y-a_k).
\end{equation}
Equations (\ref{P11}), (\ref{P12}), and (\ref{delta}) give us
\begin{equation}\label{FormulaFor[MD,K_M](x,y)}
\left[MD,K_M\right](x,y)=(-1)^{M+1}\phi_0(x)\psi_M(y)\chi_J(y)-\sum\limits_{k=1}^{2m}(-1)^ka_kK_M(x,a_k)\delta(y-a_k).
\end{equation}
Using identity (\ref{SimpleOperatorIdentity}) together with formula (\ref{FormulaFor[MD,K_M](x,y)}) we find the kernel of $[MD, (1-K_M)^{-1}]$.
Namely,
\begin{equation}\label{CommutatorMD(1-K)}
\left[MD, (1-K_M)^{-1}\right](x,y)=(-1)^{M+1}Q_0(x)\mathcal{P}^*_M(y)-\sum\limits_{k=1}^{2m}(-1)^ka_kR_M(x,a_k)\rho_M(a_k,y),
\end{equation}
where $\mathcal{P}^*_M(y)$ is defined by
$$
\mathcal{P}^*_M(y)=(1-K_M^T)^{-1}\widetilde{\psi_M}(y),\;\; \widetilde{\psi_M}(y)=\psi_M(y)\chi_J(y).
$$
Therefore,
\begin{equation}
\left[MD, (1-K_M)^{-1}\right]\phi_j(x)=(-1)^{M+1}Q_0(x)\left(\mathcal{P}_M^*,\phi_j\right)-\sum\limits_{k=1}^{2m}(-1)^ka_kR_M(x,a_k)\mathcal{Q}_j(a_k),
\end{equation}
where
$$
\left(\mathcal{P}_M^*,\phi_j\right)=\int\limits_0^{\infty}\mathcal{P}_M^*(x)\phi_j(x)dx.
$$
Note that
\begin{equation}\label{Pzvezda}
\mathcal{P}_M^*(y)=\mathcal{P}_M(y)\chi_J(y),
\end{equation}
as it can be see from the very definition of the operators $K_M^T$ and $K_M'$.
Thus we arrive to the formula
\begin{equation}\label{[MD,(I-K)]phi}
\begin{split}
\left[MD, (1-K_M)^{-1}\right]\phi_j(x)&=(-1)^{M+1}Q_0(x)V_{j,M}(a_1,\ldots,a_{2m})\\
&-\sum\limits_{k=1}^{2m}(-1)^ka_kR_M(x,a_k)\mathcal{Q}_j(a_k),
\end{split}
\end{equation}
which holds true for all $0\leq j\leq M$.

Now let us compute $(1-K_M)^{-1}MD\phi_j(x)$. Equation (\ref{PHIPSI}) implies
\begin{equation}\label{(1-K)MDphi}
(1-K_M)^{-1}MD\phi_j(x)=-\mathcal{Q}_{j+1}(x),\;\;0\leq j\leq M-1.
\end{equation}
Thus equation (\ref{PartialForQ1}) in the statement of the Proposition follows from equations (\ref{MDQj}), (\ref{[MD,(I-K)]phi})
and (\ref{(1-K)MDphi}).

Let us consider the case when $j=M$. The definition of functions $\phi_j(x)$ (see equation (\ref{PHIPSI})) implies
$$
MD\phi_M(x)=-\left(x\frac{d}{dx}\right)^{M+1}f(x).
$$
Equation (\ref{DiffEquationForf(x)}) can be written as
$$
\left(x\frac{d}{dx}\right)^{M+1}f(x)=(-1)^Mx\phi_0(x)+\sum\limits_{k=0}^M(-1)^{M-k}e_{M+1-k}(\nu_0,\ldots,\nu_M)\phi_k(x),
$$
so
$$
MD\phi_M(x)=(-1)^{M+1}x\phi_0(x)+\sum\limits_{k=0}^M(-1)^{M-k+1}e_{M+1-k}(\nu_0,\ldots,\nu_M)\phi_k(x),
$$
and
\begin{equation}\label{(1-K)MDphiM(x)}
\begin{split}
&(1-K_M)^{-1}MD\phi_M(x)=(-1)^{M+1}(1-K_M)^{-1}M\phi_0(x)\\
&+\sum\limits_{k=0}^M(-1)^{M-k+1}e_{M+1-k}(\nu_0,\ldots,\nu_M)\mathcal{Q}_k(x).
\end{split}
\end{equation}
It remains to compute $(1-K_M)^{-1}M\phi_0(x)$. We have
\begin{equation}\label{(1-K)Mphi0(x)}
\begin{split}
&(1-K_M)^{-1}M\phi_0(x)=[(1-K_M)^{-1},M]\phi_0(x)+M(1-K_M)^{-1}\phi_0(x)\\
&=[(1-K_M)^{-1},M]\phi_0(x)+x\mathcal{Q}_0(x).
\end{split}
\end{equation}
Moreover, using the formula
$$
\left[(1-K_M)^{-1},M\right]=(1-K_M)^{-1}[K_M,M](1-K_M)^{-1},
$$
and the fact that
$$
[K_M,M](x,y)=-\sum\limits_{j=0}^M\phi_j(x)\psi_j(y)\chi_J(y),
$$
we obtain
$$
\left[(1-K_M)^{-1},M\right](x,y)=-\sum\limits_{j=0}^M\mathcal{Q}_j(x)\mathcal{P}_j(y)\chi_J(y).
$$
So,
\begin{equation}\label{[(1-K)M]phi_0(x)}
[\left(1-K_M\right)^{-1},M]\phi_0(x)=-\sum\limits_{j=0}^M\mathcal{Q}_j(x)V_{0,j}(a_1,\ldots,a_{2m}).
\end{equation}
Inserting (\ref{(1-K)Mphi0(x)}) and (\ref{[(1-K)M]phi_0(x)}) into equation (\ref{(1-K)MDphiM(x)}) gives
\begin{equation}\label{(1-K)MDphiM(x)Finite}
\begin{split}
&(1-K_M)^{-1}MD\phi_M(x)=(-1)^{M+1}x\mathcal{Q}_0(x)+(-1)^M\sum\limits_{j=0}^M\mathcal{Q}_j(x)V_{0,j}(a_1,\ldots,a_{2m})\\
&+\sum\limits_{k=0}^M(-1)^{M-k+1}e_{M+1-k}(\nu_0,\ldots,\nu_M)\mathcal{Q}_k(x).
\end{split}
\end{equation}
Equation (\ref{PartialForQ2}) in the statement of the Proposition follows from
equations (\ref{MDQj}), (\ref{[MD,(I-K)]phi}), and (\ref{(1-K)MDphiM(x)Finite}).

It remains to derive equation (\ref{PartialAFromQj(x)}). We have
\begin{equation}
\frac{\partial}{\partial a_k}\mathcal{Q}_j(x;a_1,\ldots,a_{2m})=\frac{\partial}{\partial a_k}(1-K_M)^{-1}\phi_j(x).
\end{equation}
Now,
$$
\frac{\partial}{\partial a_k}(1-K_M)^{-1}=(1-K_M)^{-1}\frac{\partial K_M}{\partial a_k}(1-K_M)^{-1},
$$
and
$$
\frac{\partial K_M}{\partial a_k}(x,y)=(-1)^kK_M(x,a_k)\delta(y-a_k).
$$
Thus the kernel of $\frac{\partial}{\partial a_k}(1-K_M)^{-1}$ can be written as
\begin{equation}\label{Partiala(1-K)}
\frac{\partial}{\partial a_k}(1-K_M)^{-1}(x,y)=(-1)^kR_M(x,a_k)\rho_M(a_k,y),
\end{equation}
and we arrive to equation (\ref{PartialAFromQj(x)}).

\end{proof}
\begin{prop}\label{PropositionMDP}
The partial differential equations for the functions $\mathcal{P}_j(y;a_1,\ldots,a_{2m})$ defined by equation (\ref{FunctionPjDefinition}) can be written as follows.
For $j=0$ we have
\begin{equation}\label{PartialP1}
\begin{split}
&y\frac{\partial}{\partial y}\mathcal{P}_0(y;a_1,\ldots,a_{2m})=(-1)^{M+1}V_{0,0}(a_1,\ldots,a_{2m})\mathcal{P}_M(y;a_1,\ldots,a_{2m})\\
&+(-1)^My\mathcal{P}_M(y;a_1,\ldots,a_{2m})+(-1)^M\sum\limits_{j=0}^MV_{j,M}(a_1,\ldots,a_{2m})\mathcal{P}_j(y;a_1,\ldots,a_{2m})\\
&+\sum\limits_{k=1}^{2m}(-1)^{k-1}a_kR_M'(y,a_k)\mathcal{P}_0(a_k;a_1,\ldots,a_{2m}),
\end{split}
\end{equation}
and for $1\leq j\leq M$ we have
\begin{equation}\label{PartialP2}
\begin{split}
&y\frac{\partial}{\partial y}\mathcal{P}_j(y;a_1,\ldots,a_{2m})=\mathcal{P}_{j-1}(y;a_1,\ldots,a_{2m})\\
&+(-1)^{M-j}e_{M-j+1}(\nu_1,\ldots,\nu_M)\mathcal{P}_M(y;a_1,\ldots,a_{2m})\\
&+(-1)^{M+1}V_{0,j}(a_1,\ldots,a_{2m})\mathcal{P}_M(y;a_1,\ldots,a_{2m})\\
&
+\sum\limits_{k=1}^{2m}(-1)^{k-1}a_kR_M'(y,a_k)\mathcal{P}_j(a_k;a_1,\ldots,a_{2m}).
\end{split}
\end{equation}
In addition, for all $k$, $1\leq k\leq 2m$, we have
\begin{equation}\label{PartialP3}
\frac{\partial}{\partial a_k}\mathcal{P}_j(y;a_1,\ldots,a_{2m})=(-1)^kR_M'(y,a_k)\mathcal{P}_j(a_k;a_1,\ldots,a_{2m}).
\end{equation}
\end{prop}
\begin{proof}
We have
\begin{equation}\label{MDP1}
MD\mathcal{P}_j(x)=[MD,(1-K_M')^{-1}]\psi_j(x)+(1-K_M')^{-1}MD\psi_j(x),\;\ 0\leq j\leq M.
\end{equation}
Calculations similar to that leading to formula (\ref{[MD,(I-K)]phi}) give us
\begin{equation}\label{MDP2}
\begin{split}
[MD,(1-K_M')^{-1}]\psi_j(x)&=(-1)^{M+1}\mathcal{P}_M(x)V_{0,j}(a_1,\ldots,a_{2m})\\
&+\sum\limits_{k=1}^{2m}(-1)^{k-1}a_kR_M'(x,a_k)\mathcal{P}_j(a_k),\;\; 0\leq j\leq M.
\end{split}
\end{equation}
Note that equations (\ref{DiffEquationForg(y)}) and (\ref{PHIPSI}) imply
\begin{equation}\label{MDP3}
MD\psi_j(y)=\left\{
              \begin{array}{ll}
                \psi_{j-1}(y)-\alpha_{j-1}\psi_M(y), & 1\leq j\leq M, \\
                (-1)^My\psi_M(y), & j=0.
              \end{array}
            \right.
\end{equation}
It follows immediately  from equation (\ref{MDP3})  that
\begin{equation}\label{MDP4}
(1-K_M')^{-1}MD\psi_j(y)=\mathcal{P}_{j-1}(y)+(-1)^{M-j}e_{M-j+1}(\nu_1,\ldots,\nu_M)\mathcal{P_M}(y),
\end{equation}
where $1\leq j\leq M$. Equations (\ref{MDP2}) and (\ref{MDP4}) give equation (\ref{PartialP2})
in the statement of the Proposition. Moreover,
\begin{equation}\label{MDP5}
\begin{split}
&(1-K_M')^{-1}MD\psi_0(y)=(-1)^M(1-K_M')^{-1}M\psi_M(y)\\
&=(-1)^My\mathcal{P}_M(y)+(-1)^M\left[(1-K_M')^{-1},M\right]\psi_M(y).
\end{split}
\end{equation}
After straightforward calculations (similar to those in the proof of Proposition \ref{PropositionMDQ})
we obtain
$$
[K_M',M](x,y)=\sum\limits_{j=0}^M\psi_j(x)\phi_j(y)\chi_J(y),
$$
and
\begin{equation}\label{MDP7}
\begin{split}
&\left[(1-K_M')^{-1},M\right]\psi_M(y)=(1-K_M')^{-1}[K_M',M](1-K_M')^{-1}\psi_M(y)\\
&=\sum\limits_{j=0}^M\mathcal{P}_j(y)V_{j,M}(a_1,\ldots,a_{2m}).
\end{split}
\end{equation}
Equation (\ref{PartialP1}) follows from equations (\ref{MDP1}), (\ref{MDP2}), (\ref{MDP5}), and (\ref{MDP7}).
In order to see that equation (\ref{PartialP3}) holds true use the formula
$$
\frac{\partial}{\partial a_k}(1-K_M')^{-1}=(1-K_M')^{-1}\frac{\partial K_M'}{\partial a_k}(1-K_M)^{-1}.
$$
The kernel of $\frac{\partial}{\partial a_k}K_M'$ is
$$
(-1)^kK_M(a_k,x)\delta(y-a_k).
$$
Taking into account the first equation in (\ref{R'}) we obtain equation (\ref{PartialP3}).
\end{proof}
\begin{prop}\label{PropositionPV} We have
\begin{equation}\label{PartialV}
\frac{\partial}{\partial a_l}V_{i,j}(a_1,\ldots,a_{2m})=(-1)^l\mathcal{Q}_i(a_l;a_1,\ldots,a_{2m})\mathcal{P}_j(a_l;a_1,\ldots,a_{2m}),
\end{equation}
for $0\leq i,j\leq M$, and for $1\leq l\leq 2m$.
\end{prop}
\begin{proof}
We note that
$$
\frac{\partial}{\partial a_l}\chi_J(x)=(-1)^l\delta(x-a_l).
$$
Taking into account this equation, together with equation (\ref{PartialP3}), we obtain
\begin{equation}\label{PV3}
\begin{split}
&\frac{\partial}{\partial a_l}V_{i,j}(a_1,\ldots,a_{2m})=(-1)^l\phi_i(a_l)\mathcal{P}_j(a_l;a_1,\ldots,a_{2m})\\
&+(-1)^l\left(\int\limits_J\phi_i(x)R_M'(x,a_l)dx\right)\mathcal{P}_j(a_l;a_1,\ldots,a_{2m}).
\end{split}
\end{equation}
Now, from equation (\ref{R'}) we see that the kernel of $R_M'$,
$R_M'(x,a_l)$, is
$$
\delta(x-a_l)+\left(1-K_M'\right)^{-1}(x,a_l).
$$
Using this fact, and the fact  that the kernel of $K_M'$ (at points
$x,y$) is $K_M(y,x)\chi_J(y)$  we obtain equation (\ref{PartialV})
from equation (\ref{PV3}).
\end{proof}
\subsection{Explicit formulae for the kernels $R_M(x,y)$ and $R_M'(x,y)$}
\begin{prop}\label{PropositionKernelsRmRmPrimeExplicitFormulae}
Let $R_M$ be the resolvent of $K_M$, and $R_M'$ be the resolvent of $K_M'$. Denote by $R_M(x,y)$ the kernel of $R_M$, and denote by $R_M'(x,y)$ the kernel of $R_M'$.
We have
$$
R_M(x,y)=\frac{\sum\limits_{j=0}^M\mathcal{Q}_j(x;a_1,\ldots,a_{2m})\mathcal{P}_j(y;a_1,\ldots,a_{2m})}{x-y}\chi_J(y),
$$
and
$$
R_M'(x,y)=-\frac{\sum\limits_{j=0}^M\mathcal{P}_j(x;a_1,\ldots,a_{2m})\mathcal{Q}_j(y;a_1,\ldots,a_{2m})}{x-y}\chi_J(y).
$$
\end{prop}
\begin{proof}
Recall that the operators $R_M$ and $R_M'$ are defined by equations (\ref{R}), (\ref{R'}) correspondingly, where
$K_M$ is the operator with the kernel
\begin{equation}\label{KM(x,y)}
K_M(x,y)\chi_J(y)=\frac{\sum\limits_{j=0}^M\phi_j(x)\psi_j(y)}{x-y}\chi_J(y),
\end{equation}
and where $K_M'$ is the operator with the kernel
\begin{equation}\label{KM'(x,y)}
K_M(y,x)\chi_J(y)=-\frac{\sum\limits_{j=0}^M\psi_j(x)\phi_j(y)}{x-y}\chi_J(y).
\end{equation}
The kernels just written above are integrable in the sense of Its,
 Isergin,  Korepin and Slavnov \cite{Its}. This implies that the corresponding resolvent kernel has the same integrable form.
More explicitly,
\begin{equation}\label{R(x,y)}
R_M(x,y)=\frac{\sum\limits_{j=0}^M\left((1-K_M)^{-1}\phi_j\right)(x)\left((1-K_M^T)^{-1}\psi_j\chi_J\right)(y)}{x-y},
\end{equation}
and
\begin{equation}\label{R'(x,y)}
R_M'(x,y)=-\frac{\sum\limits_{j=0}^M\left((1-K_M')^{-1}\psi_j\right)(x)\left((1-(K_M')^T)^{-1}\phi_j\chi_J\right)(y)}{x-y}.
\end{equation}
Using equations (\ref{KM(x,y)}) and (\ref{KM'(x,y)}) we find
\begin{equation}\label{(1-K)psihi}
\left((1-K_M^T)^{-1}\psi_j\chi_J\right)(y)=\mathcal{P}_j(y;a_1,\ldots,a_{2m})\chi_J(y),
\end{equation}
and
\begin{equation}\label{(1-K)phihi}
\left((1-(K_M')^T)^{-1}\phi_j\chi_J\right)(y)=\mathcal{Q}_j(y;a_1,\ldots,a_{2m})\chi_J(y).
\end{equation}
Now, equations (\ref{R(x,y)}), (\ref{FunctionQjDefinition}) and (\ref{(1-K)psihi}) give the formula for $R_M(x,y)$
in the statement of the Proposition. In a similar way, we obtain the formula for $R_M'(x,y)$ from equations (\ref{R'(x,y)}),
(\ref{FunctionPjDefinition}), and (\ref{(1-K)phihi}).
\end{proof}
\begin{prop}\label{PropositionOneForm}
Define the $1$-form $w(a_1,\ldots,a_{2m})$ by the formula
$$
w(a_1,\ldots,a_{2m})=\sum\limits_{j=1}^{2m}\frac{\partial}{\partial a_j}\left(\log\det\left(1-K_M\right)\right)da_j.
$$
We have
$$
w(a_1,\ldots,a_{2m})=\sum\limits_{j=1}^{2m}(-1)^{j-1}R_M(a_j,a_j)da_j.
$$
\end{prop}
\begin{proof} Use the well-known formula
$$
\frac{\partial}{\partial a_j}\left(\log\det\left(1-K_M\right)\right)=-\Tr\left((1-K_M)^{-1}\frac{\partial}{\partial a_j}K_M\right),
$$
and observe that the kernel of $(1-K_M)^{-1}\frac{\partial}{\partial a_j}K_M$ is
$$
(-1)^jR_M(x,a_j)\delta(y-a_j).
$$
\end{proof}
\begin{prop}
The kernel of $R_M$, $R_M(x,y)$, satisfies the following partial differential equation
\begin{equation}\label{PDEForRM}
\begin{split}
&\left(x\frac{\partial}{\partial x}+y\frac{\partial}{\partial y}+\sum\limits_{k=1}^{2m}a_k\frac{\partial}{\partial a_k}+I\right)R_M(x,y)\\
&=(-1)^{M+1}\mathcal{Q}_0(x;a_1,\ldots,a_{2m})
\mathcal{P}_M(y;a_1,\ldots,a_{2m})\chi_J(y).
\end{split}
\end{equation}
\end{prop}
\begin{proof}
Formulae (\ref{MDL}), (\ref{CommutatorMD(1-K)}), and (\ref{Pzvezda}) give us
\begin{equation}
\begin{split}
&\left(x\frac{\partial}{\partial x}+y\frac{\partial}{\partial y}+I\right)R_M(x,y)\\
&=(-1)^{M+1}\mathcal{Q}_0(x;a_1,\ldots,a_{2m})
\mathcal{P}_M(y;a_1,\ldots,a_{2m})\chi_J(y)-\sum\limits_{k=1}^{2m}(-1)^ka_kR_M(x,a_k)\rho_M(a_k,y).
\end{split}
\nonumber
\end{equation}
Moreover,
$$
\sum\limits_{k=1}^{2m}a_k\frac{\partial}{\partial a_k}R_M(x,y)=\sum\limits_{k=1}^{2m}(-1)^ka_kR_M(x,a_k)\rho_M(a_k,y),
$$
as it follows from equation (\ref{Partiala(1-K)}).
\end{proof}
\subsection{Proof of Proposition \ref{Proposition Partial Differential Equations for x y xi eta}}
In order to obtain  partial differential equations (\ref{DynamicalEquation(a)})-(\ref{DynamicalEquation(f)}) in Proposition \ref{Proposition Partial Differential Equations for x y xi eta} we note that
$$
x_j^{(2l)}=\sqrt{-1}\mathcal{Q}_j\left(a_{2l};a_1,\ldots,a_{2m}\right),
$$
$$
y_j^{(2l)}=\sqrt{-1}\mathcal{P}_j\left(a_{2l};a_1,\ldots,a_{2m}\right),
$$
$$
x_j^{(2l-1)}=\mathcal{Q}_j\left(a_{2l-1};a_1,\ldots,a_{2m}\right),
$$
$$
y_j^{(2l-1)}=\mathcal{P}_j\left(a_{2l-1};a_1,\ldots,a_{2m}\right),
$$
and that
$$
\xi_j=(-1)^MV_{0,j}(a_1,\ldots,a_{2m})+(-1)^{M+1-j}e_{M+1-j}\left(\nu_0,\ldots,\nu_M\right),
$$
$$
\eta_j=(-1)^MV_{j,M}\left(a_1,\ldots,a_{2m}\right).
$$
In these formulae $0\leq j\leq M$, and $1\leq l\leq m$. Recall that
$\nu_0=0$. Now use the results of Propositions \ref{PropositionMDQ},
\ref{PropositionMDP}, \ref{PropositionPV}, together with the
explicit formulae for the kernels $R_M(x,y)$ and $R_M'(x,y)$
obtained in Proposition
\ref{PropositionKernelsRmRmPrimeExplicitFormulae}. This will give
formulae in Proposition \ref{Proposition Partial Differential
Equations for x y xi eta}. For example, let us check equation
(\ref{DynamicalEquation(e)}). Taking into account the definition of
$\xi_j$ (see the beginning of Section \ref{SectionPartial
differential equations}), and Proposition \ref{PropositionPV} we
obtain
$$
(-1)^M\frac{\partial}{\partial
a_l}\xi_j=(-1)^l\left(1-K_M\right)^{-1}\phi_0(a_l)\left(1-K_M'\right)^{-1}\psi_j(a_l).
$$
Now, assume that $l$ is odd. Then the above equation can be
rewritten as
$$
(-1)^M\frac{\partial}{\partial a_l}\xi_j=(-1)x_0^{(l)}y_j^{(l)}.
$$
Assume that $l$ is even. Then
$$
(-1)^M\frac{\partial}{\partial
a_l}\xi_j=(\sqrt{-1})^2x_0^{(l)}y_j^{(l)},
$$
where $\sqrt{-1}$ comes from the definition of $x_0^{(l)}$ and
$y_j^{(l)}$ for an even $l$. In both cases we obtain equation
(\ref{DynamicalEquation(e)}).

\qed
\subsection{Proof of Proposition \ref{PropositionHamiltonianEquations}, Proposition \ref{PropositionInvolution} and of Proposition \ref{PropositionSchlesinger}}
To check equations (\ref{xdot1,ydot1}) and (\ref{zetadot1,etadot1}) we need an explicit formula for the Hamiltonians $H_l$.
From the very definition of $H_l$ (equation (\ref{HlDefinition})), and from Proposition \ref{PropositionOneForm} it follows that
$$
H_l=(-1)^la_lR_M(a_l,a_l),\;\; 1\leq l\leq 2m.
$$
The quantity $R_M(a_l,a_l)$ can be obtained using formulas derived in Propositions \ref{PropositionMDQ}, \ref{PropositionMDP},
and in  Proposition \ref{PropositionKernelsRmRmPrimeExplicitFormulae}. The result is given by formula (\ref{Hl}).

Using equation (\ref{Hl}) we can check that  equations (\ref{xdot1,ydot1}) and (\ref{zetadot1,etadot1}) are in fact equivalent to
the partial differential equations (\ref{DynamicalEquation(a)})-(\ref{DynamicalEquation(f)}) in Proposition \ref{Proposition Partial Differential Equations for x y xi eta}.
This proves Proposition \ref{PropositionHamiltonianEquations}.
Equation (\ref{Involution}) of Proposition \ref{PropositionInvolution} can be checked by direct computations using equations (\ref{Hl}) and (\ref{Poisson Brackets}).  In addition, it can be checked by simple
calculations that equations
(\ref{Schlesinger1})-(\ref{Schlesinger3}) follow from
(\ref{DynamicalEquation(a)})-(\ref{DynamicalEquation(f)}).
 The equivalence of the formula  for
$H_l$ in Proposition \ref{PropositionSchlesinger} and of equation
(\ref{Hl}) can be  verified  directly as well.
 Proposition \ref{PropositionSchlesinger} is proved.
 \qed
 \subsection{Proof of Proposition \ref{PropositionNonlinearOrdinaryDiffEquation}}
 Equations (\ref{NonlinearOrdinaryDiffEquation1}) -(\ref{NonlinearOrdinaryDiffEquation6})
 is just a specialization of the general
 partial differential equations (\ref{JMMS3})-(\ref{DynamicalEquation(f)})
 in Proposition \ref{Proposition Partial Differential Equations for x y xi eta}
 to the case of the single interval $(0,s)$.
 Moreover, equation (\ref{PDEForRM}) implies
$$
\left(sR_M(s)\right)'=(-1)^{M+1}x_0(s)y_M(s),\;\; R_M(s):=R_M(s,s),
$$
or
$$
R_M(t)=\frac{(-1)^{M+1}}{t}\int\limits_0^tx_0(\tau)y_M(\tau)d\tau.
$$
On the other hand, we have
$$
\frac{d}{ds}\log\det\left(1-K_M\right)=-R_M(s),
$$
so
$$
\det\left(1-K_M\right)=\exp\left((-1)^M\int\limits_0^s\frac{1}{t}\left(\int\limits_0^tx_0(\tau)y_M(\tau)d\tau\right)dt\right).
$$
The formula for $F_M(s)$ in the statement of  Proposition \ref{PropositionNonlinearOrdinaryDiffEquation} follows from the formula just written above by integration by parts.
 Proposition \ref{PropositionNonlinearOrdinaryDiffEquation} is proved.
 \qed
\section{Conclusions}\label{SectionConclusions}
We considered a determiantal point process which arises in the
asymptotic analysis of products of rectangular Ginibre random
matrices, and describes the limiting behavior  of singular values of
such products. Our interest was to apply methods developed in Tracy
and Widom \cite{TracyWidomIntroduction}-\cite{TracyWidomSystems} to
this determinantal point process, and to describe  the gap
probabilities in terms of solutions of nonlinear differential
equations. As a result we have obtained a system of nonlinear
partial differential equations (see Proposition \ref{Proposition
Partial Differential Equations for x y xi eta}), and have
constructed a Hamiltonian system (see Proposition
\ref{PropositionHamiltonianEquations}) where the end points of
intervals play a role of multi-time independent variables. In this
construction the relevant Hamiltonians are expressible in terms of
logarithmic derivatives of the gap probabilities. The constructed
Hamiltonian system leads to the Schlesinger type equations (see
Proposition \ref{PropositionSchlesinger}), and to the isomonodromy
deformation equation of the Jimbo, Miwa, M$\hat{\mbox{o}}$ri, Ueno
and Sato theory (see Proposition
\ref{PropositionIsomonodromyFormulation} and equation
(\ref{IsomonodromyDeformationEquation})).

One of the most important problem of Random Matrix Theory is to
understand the decay of gap probabilities in the case of a single
interval $(0,s)$, as $s\rightarrow\infty$. For the determinantal
point process under considerations, this paper gives a formula for
such gap probabilities in terms of a solution of a system of
nonlinear differential equations, see  Proposition
\ref{PropositionNonlinearOrdinaryDiffEquation}. We expect that the
formula for the gap probabilities  in Proposition
\ref{PropositionNonlinearOrdinaryDiffEquation} will be useful in a
subsequent asymptotic analysis.

Finally, let us mention that similar methods can be applied to determinantal point processes on the real line formed by singular values of products of
finite independent Ginibre matrices.


\begin{thebibliography}{99}
\bibitem{AdhikariReddyReddySaha}
Adhikari, K.; Reddy, N. K.; Reddy, T. R.; Saha, K.
Determinantal point processes in the plane from products of random matrices.
 arXiv:1308.6817
\bibitem{AdlerShiotaMoerbeke}
 Adler, M.; Shiota, T.; van Moerbeke, P. Random matrices, vertex operators and
 the Virasoro algebra. Phys. Lett. A 208 (1995), no. 1-2, 67–-78.
\bibitem{AdlerShiotaMoerbeke1}
 Adler, M.; Shiota, T.; van Moerbeke, P. Random matrices,
 Virasoro algebras, and noncommutative KP. Duke Math. J. 94 (1998), no. 2, 379–-431.
\bibitem{Akemann1}
Akemann, G.; Burda, Z. Universal microscopic correlation functions for products of independent Ginibre matrices.
J. Phys. A: Math. Theor. 45 (2012),  465201.
\bibitem{AkemannKieburgWei}
Akemann, G.; Kieburg M.; Wei, L.
Singular value correlation functions for products of Wishart random matrices.
J. Phys. A. 46 (2013), 275205.
\bibitem{AkemannIpsenKieburg}
Akemann, G.; Ipsen, J.; Kieburg, M.
Products of rectangular random matrices: singular values and progressive scattering. Phys. Rev. E 88 (2013), 052118.
\bibitem{AkemannStrahov}
Akemann, G.; Strahov, E. Hole probabilities and overcrowding estimates for  products of complex Gaussian matrices.
J. Stat. Phys. 151 (2013),  987--1003.
\bibitem{Bertola}
Bertola, M.; Gekhtman, M.; Szmigielski , J. Cauchy-Laguerre two-matrix model and the Meijer-G random point field.
Comm. Math. Phys. 326 (2014), no. 1, 111--144.
\bibitem{Borodin}
Borodin, A. Biorthogonal ensembles. Nuclear Phys. B 536 (1999), no. 3, 704--732.
\bibitem{DeiftIntegrableOperators}
Deift, P. A. Integrable operators. In: V.~Buslaev, M.~Solomyak,
D.~Yafaev (eds) Differential operators and spectral theory:
M.~Sh.~Birman's 70th anniversary collection. American Mathematical
Society Translations, ser. 2, 189, Providence, R.I., AMS, (1999).
\bibitem{FokasItsKapaevNovokshenovBook}
Fokas, A. S.; Its, A. R.; Kapaev, A. A.; Novokshenov, V. Y. Painlev$\acute{\mbox{e}}$ transcendents. The Riemann-Hilbert approach. Mathematical Surveys and Monographs, 128. American Mathematical Society, Providence, RI, 2006.
\bibitem{ForresterLogGases}
 Forrester, P. J. Log-gases and random matrices. London Mathematical Society Monographs Series, 34. Princeton University Press, Princeton, NJ, 2010.
\bibitem{Forrester}
Forrester, P. J. Eigenvalue statistics for product  complex Wishart
matrices. arXiv: 1401.2572v1.
 \bibitem{Harnad}
  Harnad, J. On the bilinear equations for Fredholm determinants appearing in random matrices. J. Nonlinear Math. Phys. 9 (2002), no. 4, 530–-550.
\bibitem{HarnadTracyWidom}
 Harnad, J.; Tracy, C. A.; Widom, H. Hamiltonian structure of equations appearing in random matrices. Low-dimensional topology and quantum field theory (Cambridge, 1992), 231–-245, NATO Adv. Sci. Inst. Ser. B Phys., 315, Plenum, New York, 1993.
\bibitem{IpsenKieburg}
 Ipsen, J.; Kieburg, M.  Weak Communication Relations and Eigenvalue Statistics for Products of Rectangular Random Matrices. Phys. Rev. E 89, (2014), 032106.
\bibitem{Its}
Its, A. R.; Isergin, A. G.; Korepin, V. E.;  Slavnov, N. A.
Differential equations for quantum correlation functions. Intern. J.
Mod. Phys., B4 (1990), 1003--1037.
\bibitem{Its1}
 Its, A. Painlev$\acute{\mbox{e}}$ transcendents. The Oxford handbook of random matrix theory, 176–-197, Oxford Univ. Press, Oxford, 2011.
\bibitem{ItsHarnad}
Its, A.; Harnad, J. Integrable Fredholm operators and dual isomonodromic deformations. Comm. Math. Phys. 226 (2002), no. 3, 497--530.
 \bibitem{Jimbo}
  Jimbo, M.; Miwa, T.; M$\hat{\mbox{o}}$ri, Y.; Sato, M. Density matrix of an impenetrable Bose gas and the fifth Painlev$\acute{\mbox{e}}$ transcendent. Phys. D 1 (1980), no. 1, 80–-158.
 \bibitem{Jimbo1}
  Jimbo, M.; Miwa, T.; Ueno, K. Monodromy preserving deformation of linear ordinary differential equations with rational coefficients. I. General theory and $\tau$-function. Phys. D 2 (1981), no. 2, 306–-352.
 \bibitem{Jimbo2}
 Jimbo, M.; Miwa, T. Monodromy preserving deformation of linear ordinary differential equations with rational coefficients. II. Phys. D 2 (1981), no. 3, 407–-448.
\bibitem{Kuijlaars} Kuijlaars, A. B. J.;  Zhang, L.
Singular values of products of Ginibre random matrices, multiple orthogonal polynomials and hard edge scaling limits.
arXiv:1308.1003
\bibitem{Luke}
Luke, Y. L.  The special functions and their approximations. Academic Press, New York 1969.
\bibitem{Moerbeke1}
 van Moerbeke, P. Random matrix theory and integrable systems. The Oxford handbook of random matrix theory, 198–230, Oxford Univ. Press, Oxford, 2011.
\bibitem{Moerbeke2}
van Moerbeke, P. Random and integrable models in mathematics and physics. Random matrices, random processes and integrable systems, 3–130, CRM Ser. Math. Phys., Springer, New York, 2011.
\bibitem{Muller1}
M$\ddot{\mbox{u}}$ller, R. R. On the asymptotic eigenvalue distribution of concatenated vectorvalued fading channels. IEEE Trans. Inf. Theor. 48 (2002) 2086--2091.
\bibitem{Muller2}
M$\ddot{\mbox{u}}$ller, R. R. Free probability, Ch. 5 in `` Random matrix theory for wireless communications'', online resourse, 2007.
\bibitem{Olver}
Olver, F. W. J.; Lozier, R. F.; Boisvert, R. F.; Clark, C. W., editors. NIST Handbook of Mathematical Functions, Cambridge University Press, Cambridge 2010.
\bibitem{Palmer}
Palmer, J. Deformation analysis of matrix models. Phys. D 78 (1994), no. 3--4, 166–-185.
\bibitem{TracyWidomIntroduction}
Tracy, C. A.; Widom, H. Introduction to random matrices. Geometric and quantum aspects of integrable systems (Scheveningen, 1992), 103–130, Lecture Notes in Phys., 424, Springer, Berlin, 1993.
\bibitem{TracyAiry}
 Tracy, C. A.; Widom, H. Level-spacing distributions and the Airy kernel. Comm. Math. Phys. 159 (1994), no. 1, 151–-174
\bibitem{Tracy}
 Tracy, C. A.; Widom, H. Level spacing distributions and the Bessel kernel. Comm. Math. Phys. 161 (1994), no. 2, 289–-309.
 \bibitem{TracyWidom}
Tracy, C. A.;  Widom, H. Fredholm determinants, differential
equations and matrix models.
Commun. Math. Phys. 163 (1994), 33–-72.
 \bibitem{TracyWidomSystems}
Tracy, C. A.;  Widom, H.  Systems of partial differential equations for a class of operator determinants. Partial differential operators and mathematical physics.
(Holzhau, 1995). 381--388, Oper. Theory Adv. Appl., 78. Birkh$\ddot{\mbox{a}}$user, Basel, 1995.
\bibitem{Zhang}
Zhang, L. A note on the limiting mean distribution for products of two Wishart random matrices. J. Math. Phys. 54 (2013), no. 8, 083303.
\end{thebibliography}
\end{document}